\documentclass[onecolumn, a4size, 12pt]{IEEEtran}
\usepackage{amsmath}
\usepackage{amssymb}
\usepackage{amsfonts}
\usepackage{graphicx}
\usepackage{epsfig}
\usepackage{subfigure}
\usepackage{psfrag}

\title{On Peak versus Average Interference Power Constraints for Protecting Primary Users in Cognitive Radio Networks \footnote{Submitted to
IEEE Transactions on Wireless Communications, June 1 2008, revised
September 1 2008.}}

\author{Rui Zhang \footnote{Rui Zhang is with the Institute for Infocomm Research, A*STAR, Singapore.
(e-mail: rzhang@i2r.a-star.edu.sg)}}

\begin{document}
\maketitle \maketitle \thispagestyle{empty}

\begin{abstract}
This paper considers spectrum sharing for wireless communication
between a cognitive radio (CR) link and a primary radio (PR) link.
It is assumed that the CR protects the PR transmission by applying
the so-called {\it interference-temperature} constraint, whereby the
CR is allowed to transmit regardless of the PR's on/off status
provided that the resultant interference power level at the PR
receiver is kept below some predefined threshold. For the fading PR
and CR channels, the interference-power constraint at the PR
receiver is usually one of the following two types: One is to
regulate the {\it average} interference power (AIP) over all the
fading states, while the other is to limit the {\it peak}
interference power (PIP) at each fading state. From the CR's
perspective, given the same average and peak power threshold, the
AIP constraint is more favorable than the PIP counterpart because of
its more flexibility for dynamically allocating transmit powers over
the fading states. On the contrary, from the perspective of
protecting the PR, the more restrictive PIP constraint appears at a
first glance to be a better option than the AIP. Some surprisingly,
this paper shows that in terms of various forms of capacity limits
achievable for the PR fading channel, e.g., the ergodic and outage
capacities, the AIP constraint is also superior over the PIP. This
result is based upon an interesting {\it interference diversity}
phenomenon, i.e., randomized interference powers over the fading
states in the AIP case are more advantageous over deterministic ones
in the PIP case for minimizing the resultant PR capacity losses.
Therefore, the AIP constraint results in larger fading channel
capacities than the PIP for both the CR and PR transmissions.
\end{abstract}


\setlength{\baselineskip}{1.3\baselineskip}
\newtheorem{remark}{\underline{Remark}}[section]
\newtheorem{theorem}{\underline{Theorem}}[section]
\newtheorem{lemma}{\underline{Lemma}}[section]
\newcommand{\mv}[1]{\mbox{\boldmath{$ #1 $}}}

\section{Introduction}

This paper is concerned with a typical spectrum sharing scenario for
wireless communication, where a secondary radio, also commonly known
as the {\it cognitive radio} (CR), communicates over the same
bandwidth that has been allocated to an existing primary radio (PR).
For such a scenario, the CR usually needs to deal with a fundamental
tradeoff between maximizing its own transmission throughput and
minimizing the amount of interference caused to the PR transmission.
There are in general three types of methods known in the literature
for the CR to deal with such a tradeoff. One is the so-called {\it
opportunistic spectrum access} (OSA), originally outlined in
\cite{Mitola00} and later formally introduced by DARPA, whereby the
CR decides to transmit over the PR spectrum only when the PR
transmission is detected to be off; while the other two methods
allow the CR to transmit over the spectrum simultaneously with the
PR. One of them is based on the ``cognitive relay'' idea
\cite{Tarokh06}, \cite{Viswanath06}. For this method, the CR
transmitter is assumed to know perfectly all the channels from CR/PR
transmitter to PR and CR receivers and, furthermore, the PR's
message prior to the PR transmission. Thereby, the CR transmitter is
able to send messages to its own receiver and, at the same time,
compensate for the resultant interference to the PR receiver by
operating as an assisting relay to the PR transmission. In contrast,
the other method only requires that the power gain of the channel
from CR transmitter to PR receiver is known to the CR transmitter
and, thereby, the CR is allowed to transmit regardless of the PR's
on/off status provided that the resultant interference power level
at the PR receiver is kept below some predefined threshold, also
known as the {\it interference-temperature} constraint
\cite{Haykin05}, \cite{Gastpar07}. In this paper, we focus our study
on this method due to its many advantages from an implementation
viewpoint.

To enable wireless spectrum sharing under the
interference-temperature constraint, {\it dynamic resource
allocation} (DRA) for the CR becomes crucial, whereby the transmit
power level, bit-rate, bandwidth, and antenna beam of the CR are
dynamically changed based upon the channel state information (CSI)
available at the CR transmitter. For the single-antenna fading PR
and CR channels, transmit power control for the CR has been studied
in \cite{Ghasemi07}, \cite{Aissa07} under the average/peak
interference-power constraint at the PR receiver based upon the CSI
on the channels from the CR transmitter to the CR and PR receivers,
in \cite{Kang08} under the combined interference-power constraint
and the CR's own transmit-power constraint, and in \cite{Zhang08a},
\cite{Chen07} based upon the additional CSI on the PR fading
channel. On the other hand, for the multi-antenna PR and CR
channels, in \cite{Zhang08b} the authors proposed both optimal and
suboptimal spatial adaptation schemes for the CR transmitter.
Information-theoretic limits for multiuser multi-antenna/fading CR
channels have also been studied in, e.g., \cite{ZhangLan08a},
\cite{Zhang08MAC}.

In this paper, we consider the single-antenna fading PR and CR
channels. For such scenarios, the interference-power constraint at
the PR receiver is usually one of the following two types: One is
the {\it long-term} constraint that regulates the {\it average}
interference power (AIP) over all the fading states, while the other
is the {\it short-term} one that limits the {\it peak} interference
power (PIP) at each of the fading states. Clearly, the PIP
constraint is more restrictive than the AIP counterpart given the
same average and peak interference-power threshold. From the CR's
perspective, the AIP constraint is more favorable than the PIP,
since the former provides the CR more flexibility for dynamically
allocating transmit powers over the fading states and, thus,
achieves larger fading channel capacities \cite{Aissa07},
\cite{Kang08}. However, the effect of the AIP- and PIP-based CR
power control on the PR transmissions has not yet been studied in
the literature, to the author's best knowledge. At a first glance,
the more restrictive PIP constraint seems to be a better option than
the AIP from the perspective of protecting the PR. Some
surprisingly, in this paper the contrary conclusion is rigourously
shown, i.e., the AIP constraint is indeed superior over the PIP in
terms of various forms of capacity limits achievable for the PR
fading channel, e.g., the ergodic and outage capacities. This result
is due to an interesting {\it interference diversity} phenomenon for
the PR transmission: Due to the convexity of the capacity function
with respect to the noise/interferecne power, more randomized
interference powers over the fading states at the PR receiver in the
AIP case are more advantageous over deterministic ones in the PIP
case for minimizing the resultant PR capacity losses. Therefore,
this paper provides an important design rule for the CR networks in
practice, i.e., the AIP constraint may result in improved fading
channel capacities over the PIP for both the CR and PR
transmissions.

The rest of this paper is organized as follows. Section
\ref{sec:channel model} presents the system model for spectrum
sharing. Section \ref{sec:secondary user capacity} considers the CR
link and summarizes the results known in the literature on the CR
fading channel capacities and the corresponding optimal
power-control policies under the AIP or the PIP constraint. Section
\ref{sec:primary user capacity} then studies various forms of the PR
fading channel capacities under the interference from the CR
transmitter due to the AIP- or PIP-based CR power control, and
proves that the AIP constraint results in larger channel capacities
than the PIP for the same power threshold. Section
\ref{sec:simulation results} considers both PR and CR transmissions
and shows the simulation results on their jointly achievable
capacities under spectrum sharing. Finally, Section
\ref{sec:conclusion} concludes this paper.

{\it Notation}: $|z|$ denotes the Euclidean norm of a complex number
$z$. $\mathbb{E}[\cdot]$ denotes the statistical expectation.
$Pr\{\cdot\}$ denotes the probability. ${\bf 1}(\mathcal{A})$
denotes the indicator function taking the value of one if the event
$\mathcal{A}$ is true, and the value of zero otherwise. The
distribution of a circular symmetric complex Gaussian (CSCG) random
variable (r.v.) with mean $x$ and variance $y$ is denoted as
$\mathcal{CN}(x,y)$, and $\sim$ means ``distributed as''.
$\max(x,y)$ and $\min(x,y)$ denote, respectively, the maximum and
the minimum between two real numbers $x$ and $y$; for a real number
$a$, $(a)^+\triangleq\max(0,a)$.

\section{System Model} \label{sec:channel model}

\begin{figure}
\psfrag{a}{\hspace{0.3in} PR-Tx}\psfrag{b}{\hspace{0.3in} CR-Tx}
\psfrag{c}{PR-Rx}\psfrag{d}{CR-Rx}
\psfrag{e}{$\mv{f}$}\psfrag{f}{$\mv{h}$}\psfrag{h}{$\mv{g}$}
\begin{center}
\scalebox{1.0}{\includegraphics*[88pt,617pt][302pt,755pt]{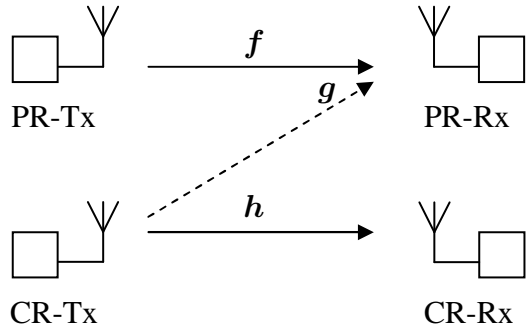}}
\end{center}
\caption{Spectrum sharing between a PR link and a CR
link.}\label{fig:system model}
\end{figure}

As shown in Fig. \ref{fig:system model}, a spectrum sharing scenario
is considered where a CR link consisting of a CR transmitter (CR-Tx)
and a CR receiver (CR-Rx) shares the same bandwidth for transmission
with an existing PR link consisting of a PR transmitter (PR-Tx) and
a PR receiver (PR-Rx). All terminals are assumed to be equipped with
a single antenna. We consider a slow-fading environment and, for
simplicity, assume a block-fading (BF) channel model for all the
channels involved in the PR-CR network. Furthermore, we assume
coherent communication and thus only the fading channel power gain
(amplitude square) is of interest. Denote $\mv{h}$ as the r.v. for
the power gain of the fading channel from CR-Tx to CR-Rx. Similarly,
$\mv{g}$ and $\mv{f}$ are defined for the fading channel from CR-Tx
to PR-Rx and PR-Tx to PR-Rx, respectively. For convenience, in this
paper we ignore the channel from PR-Tx to CR-Rx. Denote $i$ as the
joint fading state for all the channels involved. Then, let $h_i$ be
the $i$th component in $\mv{h}$ for fading state $i$; similarly,
$g_i$ and $f_i$ are defined. It is assumed that $h_i$, $g_i$, and
$f_i$ are independent of each other, and all of them have continuous
probability density functions (PDFs). It is also assumed that the
additive noises at both PR-Rx and CR-Rx are independent CSCG r.v.s
each $\sim\mathcal{CN}(0,1)$. Since we are interested in the
information-theoretic limits of the PR and CR channels, it is
assumed that the optimal Gaussian codebook is used at both PR-Tx and
CR-Tx.

For the PR link, the transmit power at fading state $i$ is denoted
as $q_i$. It is assumed that the PR is oblivious to the CR
transmission and thus does not attempt to protect the CR nor
cooperate with the CR for transmission. Due to the CR transmission,
PR-Rx may observe an additional interference power, denoted as
$I_i=g_ip_i$, at fading state $i$ where $p_i$ denotes the CR
transmit power at fading state $i$. The PR power-control policy,
denoted as $\mathcal{P}_{\rm PR}(\mv{f},\mv{I})$, is in general a
mapping from $f_i$ and $I_i$ to $q_i$ for each $i$, with $I_i$ being
the $i$th component of $\mv{I}$, subject to an average transmit
power constraint $Q$, i.e., $\mathbb{E}[q_i]\leq Q$. By treating the
interference from CR-Tx as the additional Gaussian noise at PR-Rx,
the mutual information of the PR fading channel for fading state $i$
under a given $\mathcal{P}_{\rm PR}(\mv{f},\mv{I})$ can then be
expressed as \cite{Coverbook}
\begin{equation}\label{eq:PR MI}
R_{\rm PR}(i)=\log\left(1+\frac{f_iq_i}{1+I_i}\right).
\end{equation}

For the CR link, since the CR needs to protect the PR transmission,
the CR power-control policy needs to be aware of both the PR and CR
transmissions. It is assumed that the channel power gains $g_i$ and
$h_i$ are perfectly known at CR-Tx for each $i$.\footnote{In
practice, the channel power gain between CR-Tx and PR-Rx can be
obtained at CR-Tx via, e.g., estimating the received signal power
from PR-Rx when it transmits, under the assumptions of the
pre-knowledge on the PR-Rx transmit power level and the channel
reciprocity.} Thus, the CR power-control policy can be expressed as
$\mathcal{P}_{\rm CR}(\mv{h},\mv{g})$ with $\mv{h}$ consisting of
$h_i$'s, subject to an average transmit power constraint $P$, i.e.,
$\mathbb{E}[p_i]\leq P$. The mutual information of the CR fading
channel for fading state $i$ under a given $\mathcal{P}_{\rm
CR}(\mv{h},\mv{g})$ can then be expressed as
\begin{equation}\label{eq:CR MI}
R_{\rm CR}(i)=\log\left(1+h_ip_i\right).
\end{equation}

In this paper, we assume that the CR protects the PR transmission
via transmit power control by applying the interference-power
constraint at PR-Rx, in the form of either the AIP or the PIP. The
AIP constraint regulates the average interference power at PR-Rx
over all the fading states and is thus expressed as
\begin{equation}\label{eq:AIP}
\mathbb{E}[I_i]\leq \Gamma_{a} \ \ {\rm or} \ \
\mathbb{E}[g_ip_i]\leq \Gamma_{a}
\end{equation}
where $\Gamma_a$ denotes the predefined AIP threshold. In contrast,
the PIP constraint limits the peak interference power at PR-Rx at
each of the fading states and is thus expressed as
\begin{equation}\label{eq:PIP}
I_i\leq \Gamma_{p}, \forall i \ \ {\rm or} \ \ g_ip_i\leq
\Gamma_{p}, \forall i
\end{equation}
where $\Gamma_p$ denotes the predefined PIP threshold. Note that the
PIP constraint is in general more restrictive over the AIP. This can
be easily seen by observing that given $\Gamma_p=\Gamma_a$,
(\ref{eq:PIP}) implies (\ref{eq:AIP}) but not vice versa. Therefore,
from the CR's perspective, applying the AIP constraint is more
favorable than the PIP because the former provides the CR more
flexibility for adapting transmit powers over the fading states.

In this paper, we consider two well-known capacity limits for the
fading PR and CR channels, namely, the ergodic capacity and the
outage capacity. The ergodic capacity measures the maximum average
rate over the fading states \cite{Goldsmith97}, while the resultant
mutual information for each fading state can be variable. In
contrast, the outage capacity measures the maximum constant rate
that is achievable over each of the fading states with a guaranteed
outage probability \cite{Shamai}, \cite{Caire}. In the extreme case
of zero outage probability, the outage capacity is also known as the
delay-limited capacity \cite{Tse98}. In general, the ergodic and
delay-limited capacities can be considered as the throughput limits
for a fading channel with no and with minimal transmission delay
requirement, respectively.

\section{CR Capacities under AIP versus PIP Constraint} \label{sec:secondary user
capacity}

In this section, we summarize the results known in the literature on
the CR fading channel capacities and the corresponding optimal
power-control policies under the AIP or the PIP constraint. Consider
first the AIP case. The optimal $\mathcal{P}_{\rm
CR}(\mv{h},\mv{g})$  to achieve the ergodic capacity of the CR
fading channel is expressed as \cite{Ghasemi07}
\begin{equation}\label{eq:power AIP ER}
p_i^{{\rm ER}, a}=\left(\frac{1}{\nu g_i}-\frac{1}{h_i}\right)^+
\end{equation}
where $\nu$ is a positive constant determined from
$\mathbb{E}[g_ip_i^{{\rm ER}, a}]=\Gamma_a$. Note that the above
power control resembles the well-known ``water filling (WF)'' power
control \cite{Coverbook}, \cite{Goldsmith97}, which achieves the
ergodic capacity of the conventional fading channel, whereas there
is also a key difference here: In (\ref{eq:power AIP ER}), the
so-called ``water level'' for WF, $1/(\nu g_i)$, depends on the
channel power gain $g_i$ from CR-Tx to PR-Rx as compared to being a
constant in the standard WF power control. Substituting
(\ref{eq:power AIP ER}) into $R_{\rm CR}(i)$ given in (\ref{eq:CR
MI}) and taking the expectation of the resultant $R_{\rm CR}(i)$
over $i$, we obtain the ergodic capacity for the CR under the AIP
constraint, denoted as $C_{\rm CR}^{{\rm ER},a}$. On the other hand,
the optimal $\mathcal{P}_{\rm CR}(\mv{h},\mv{g})$ to achieve the
outage capacity of the CR fading channel with a guaranteed outage
probability, $\epsilon_{0}$, is expressed as \cite{Aissa07},
\cite{Kang08}
\begin{eqnarray}\label{eq:power AIP OUT}
p_i^{{\rm OUT}, a}=\left\{ \begin{array}{ll} \frac{\zeta_a}{h_i}, &
\frac{h_i}{g_i}\geq \lambda \\ 0, & {\rm otherwise}
\end{array} \right.
\end{eqnarray}
where $\lambda$ is a nonnegative constant determined from
$Pr\{h_i/g_i< \lambda\}=\epsilon_0$, and $\zeta_a$ is the constant
signal-to-noise ratio (SNR) at CR-Rx obtained from
$\mathbb{E}[g_ip_i^{{\rm OUT}, a}]=\Gamma_a$. Note that the above
power control resembles the well-known ``truncated channel inversion
(TCI)'' power control \cite{Goldsmith97} to achieve the outage
capacity of the conventional fading channel \cite{Caire}, while
there is also a difference between (\ref{eq:power AIP OUT}) and the
standard TCI on the threshold value $\lambda$ for power truncation
(no transmission): in (\ref{eq:power AIP OUT}) $\lambda$ depends on
the ratio between $h_i$ and $g_i$, as compared to only $h_i$ in the
standard TCI. The corresponding outage capacity, denoted as $C_{\rm
CR}^{{\rm OUT},a}(\epsilon_0)$, is then obtained as
$\log(1+\zeta_a)$. Note that if $\epsilon_0=0$, it then follows that
$\lambda=0$ and the resultant power-control policy in (\ref{eq:power
AIP OUT}) becomes the ``channel inversion (CI)'' power control
\cite{Goldsmith97}, which achieves the delay-limited capacity for
the CR \cite{Caire}, denoted as $C_{\rm CR}^{{\rm DL},a}$.

Consider next the PIP case. It is easy to show that in this case the
optimal $\mathcal{P}_{\rm CR}(\mv{h},\mv{g})$ should use the maximum
possible transmit power for each fading state $i$ so as to maximize
both the ergodic and the outage capacities,\footnote{It is noted
that to achieve the same outage capacity for the CR, under the
assumption that the CR channel power gain $h_i$ is known at CR-Tx
for each $i$, it is possible for the CR power control to assign a
smaller power value $\zeta_p/h_i$ than $\Gamma_p/g_i$ if
 the former happens to be smaller than the latter for some
$i$. However, if $h_i$s are not available at CR-Tx, it is optimal
for the CR to assign the maximum possible transmit power
$\Gamma_p/g_i$ for each $i$ to minimize the outage probability.
Therefore, in this paper we consider that $p_i^{{\rm OUT},
p}=\Gamma_p/g_i, \forall i$.} thus, we have
\begin{equation}\label{eq:power PIP}
p_i^{{\rm ER}, p}=p_i^{{\rm OUT}, p}=\frac{\Gamma_p}{g_i}, \forall
i.
\end{equation}
The resultant ergodic capacity, denoted as $C_{\rm CR}^{{\rm
ER},p}$, is then obtained accordingly from (\ref{eq:CR MI}). The
resultant outage probability, $\epsilon_0$, can be shown equal to
$Pr\{(\Gamma_p/g_i)h_i<\zeta_p\}$ where $\zeta_p$ is the constant
SNR at CR-Rx. For a given $\epsilon_0$, the corresponding $\zeta_p$
can thus be obtained, as well as the corresponding outage capacity,
$C_{\rm CR}^{{\rm OUT},p}(\epsilon_0)= \log(1+\zeta_p)$. It is easy
to see that if $\epsilon_0=0$, it follows that $\zeta_p=0$ and thus
the delay-limited capacity for the CR under the PIP constraint,
denoted as $C_{\rm CR}^{{\rm DL},p}$, is always zero.

Comparing the power allocations in (\ref{eq:power AIP ER}) and
(\ref{eq:power AIP OUT}) for the AIP case with those in
(\ref{eq:power PIP}) for the PIP case, it is easy to see that the
former power allocations are more flexible than the latter ones over
the fading states. Furthermore, the AIP-based power control depends
on both the channel power gains, $h_i$ and $g_i$, while the
PIP-based power control only depends on $g_i$. As a result, under
the same average and peak power threshold, i.e.,
$\Gamma_a=\Gamma_p$, it is easy to show that $C_{\rm CR}^{{\rm
ER},a}\geq C_{\rm CR}^{{\rm ER},p}$, $C_{\rm CR}^{{\rm
OUT},a}(\epsilon_0)\geq C_{\rm CR}^{{\rm OUT},p}(\epsilon_0)$, and
$C_{\rm CR}^{{\rm DL},a}\geq C_{\rm CR}^{{\rm DL},p}$. Thus, the AIP
is superior over the PIP in terms of the fading channel capacity
limits achievable for the CR.

\section{PR Capacities under AIP versus PIP Constraint} \label{sec:primary user
capacity}

In this section, we will present the main contributions of this
paper on the comparison of the effects of AIP and PIP constraints on
various fading channel capacities for the PR. For fair comparison,
we consider the same average and peak interference-power threshold,
i.e., $\Gamma_a=\Gamma_p=\Gamma$. Note that both AIP and PIP
constraints are satisfied with equalities at PR-Rx for all the CR
power-control policies presented in Section \ref{sec:secondary user
capacity}, i.e., for the AIP case, $\mathbb{E}[I_i]=\Gamma$; and for
the PIP case, $I_i=\Gamma, \forall i$. In the following two
subsections, we consider the ergodic capacity and the outage
capacity for the PR fading channel, respectively.

\subsection{Ergodic Capacity}

\subsubsection{Constant-Power Policy} The simplest power
control for the PR is the {\it constant-power} (CP) policy, i.e.,
\begin{equation}\label{eq:CP}
q_i^{\rm CP}=Q, \ \forall i.
\end{equation}
CP is an attractive scheme in practice from an implementation
viewpoint since it does not require any CSI on the PR fading channel
at PR-Tx. In addition, CP satisfies a peak transmit-power constraint
for all the fading states. With CP, the ergodic capacity of the PR
fading channel in the AIP case can be obtained from (\ref{eq:PR MI})
and expressed as
\begin{equation}\label{eq:primary capacity AIP CP}
C^{{\rm ER},a}_{\rm
PR,CP}=\mathbb{E}\left[\log\left(1+\frac{f_iQ}{1+I_i}\right)\right]
\end{equation}
and in the PIP case expressed as
\begin{equation}\label{eq:primary capacity PIP CP}
C^{{\rm ER},p}_{\rm
PR,CP}=\mathbb{E}\left[\log\left(1+\frac{f_iQ}{1+\Gamma}\right)\right].
\end{equation}

\begin{theorem}\label{theorem:CP}
With the CP policy for the PR, $C^{{\rm ER},a}_{\rm PR,CP}\geq
C^{{\rm ER},p}_{\rm PR,CP}$, under the same average and peak power
threshold $\Gamma$.
\end{theorem}
\begin{proof}
The following equalities/inequality hold:
\begin{eqnarray*}
C^{{\rm ER},a}_{\rm
PR,CP}&\overset{(a)}{=}&\mathbb{E}_f\mathbb{E}_I\left[\log\left(1+\frac{f_iQ}{1+I_i}\right)\right]
\\ &\overset{(b)}{\geq}&
\mathbb{E}_f\left[\log\left(1+\frac{f_iQ}{1+\mathbb{E}[I_i]}\right)\right]
\\ &\overset{(c)}{=}&
\mathbb{E}_f\left[\log\left(1+\frac{f_iQ}{1+\Gamma}\right)\right]
\\ &\overset{(d)}{=}& C^{{\rm ER},p}_{\rm
PR,CP}
\end{eqnarray*}
where $(a)$ is from (\ref{eq:primary capacity AIP CP}) and due to
independence of $f_i$ and $g_i$ and thus $f_i$ and $I_i$; $(b)$ is
due to convexity of the function
$f(x)=\log\left(1+\frac{\kappa}{1+x}\right)$ where $\kappa$ is any
positive constant and $x\geq 0$, and Jensen's inequality (e.g.,
\cite{Coverbook}); $(c)$ and $(d)$ are due to
$\mathbb{E}[I_i]=\Gamma$ and (\ref{eq:primary capacity PIP CP}),
respectively.
\end{proof}

Theorem \ref{theorem:CP} suggests that, some surprisingly, the AIP
constraint that results in randomized interference power levels over
the fading states at PR-Rx is in fact more advantageous for
improving the PR ergodic capacity over the PIP constraint that
results in constant interference power levels at all the fading
states, for the same value of $\Gamma$. As shown in the above proof,
this result is mainly due to the convexity of the capacity function
with respect to the noise/interference power. We thus name this
interesting phenomenon for the PR transmission in a CR network as
``interference diversity''.

\subsubsection{Water-Filling Power Control}
If the effective channel power gain, $f_i/(1+I_i)$, for the PR
fading channel is known at PR-Tx for each $i$, the optimal
$\mathcal{P}_{\rm PR}(\mv{f},\mv{I})$ to achieve the ergodic
capacity for the PR is the standard WF power-control policy. In the
AIP case, the optimal power allocation is expressed as
\begin{equation}\label{eq:WF AIP}
q_i^{{\rm WF}, a}=\left(\frac{1}{\mu_a}-\frac{1+I_i}{f_i}\right)^+
\end{equation}
where $\mu_a$ controls the water level, $1/\mu_a$, with which
$\mathbb{E}[q_i^{{\rm WF},a}]=Q$. From (\ref{eq:WF AIP}), the
ergodic capacity for the PR in the AIP case is obtained as
\begin{equation}\label{eq:capacity WF AIP}
C^{{\rm ER},a}_{\rm
PR,WF}=\mathbb{E}\left[\left(\log\left(\frac{f_i}{\mu_a(1+I_i)}\right)\right)^+\right].
\end{equation}

Similarly, we can obtain the optimal WF-based power control for the
PR in the PIP case as
\begin{equation}\label{eq:WF PIP}
q_i^{{\rm WF},p}=\left(\frac{1}{\mu_p}-\frac{1+\Gamma}{f_i}\right)^+
\end{equation}
where $\mu_p$ is obtained from $\mathbb{E}[q_i^{{\rm WF},p}]=Q$. The
corresponding ergodic capacity then becomes
\begin{equation}\label{eq:capacity WF PIP}
C^{{\rm ER},p}_{\rm
PR,WF}=\mathbb{E}\left[\left(\log\left(\frac{f_i}{\mu_p(1+\Gamma)}\right)\right)^+\right].
\end{equation}

Next, we first show that an intuitive method to compare $C^{{\rm
ER},a}_{\rm PR,WF}$ in (\ref{eq:capacity WF AIP}) and $C^{{\rm
ER},p}_{\rm PR,WF}$ in (\ref{eq:capacity WF PIP}) does not work
here. Then, we present a different method for such comparison. One
intuitive method to compare $C^{{\rm ER},a}_{\rm PR,WF}$ and
$C^{{\rm ER},p}_{\rm PR,WF}$ would be as follows. If it can be shown
that $\mu_a<\mu_p$, then due to convexity of the function
$g(x)=\left(\log\left(\frac{\kappa}{1+x}\right)\right)^+$ where
$\kappa$ is a positive constant and $x\geq 0$, and similarly like
the proof of Theorem \ref{theorem:CP}, it can be shown that $C^{{\rm
ER},a}_{\rm PR,WF}> C^{{\rm ER},p}_{\rm PR,WF}$. Unfortunately, in
the following we prove by contradiction that the opposite inequality
is in fact true for $\mu_a$ and $\mu_p$. Thus, we can not conclude
which one of $C^{{\rm ER},a}_{\rm PR,WF}$ and $C^{{\rm ER},p}_{\rm
PR,WF}$ is indeed larger by this intuitive method.

Supposing that $\mu_a<\mu_p$, we then have
\begin{eqnarray*}
\mathbb{E}[q_i^{{\rm WF}, a}]&=&
\mathbb{E}\left[\left(\frac{1}{\mu_a}-\frac{1+I_i}{f_i}\right)^+\right]
\\ &>&
\mathbb{E}\left[\left(\frac{1}{\mu_p}-\frac{1+I_i}{f_i}\right)^+\right]
\\ &=&
\mathbb{E}_f\mathbb{E}_I\left[\left(\frac{1}{\mu_p}-\frac{1+I_i}{f_i}\right)^+\right]
\\ &\overset{(a)}{\geq}&
\mathbb{E}_f\left[\left(\frac{1}{\mu_p}-\frac{1+\mathbb{E}[I_i]}{f_i}\right)^+\right]
\\ &=&
\mathbb{E}\left[\left(\frac{1}{\mu_p}-\frac{1+\Gamma}{f_i}\right)^+\right]
\\ &=&
\mathbb{E}[q_i^{{\rm WF},p}]
\\ &=&
Q
\end{eqnarray*}
where $(a)$ is due to convexity of the function
$z(x)=\left(\kappa_1-\frac{1+x}{\kappa_2}\right)^+$ where $\kappa_1$
and $\kappa_2$ are positive constants and $x\geq 0$, and Jensen's
inequality. Since it is known that $\mathbb{E}[q_i^{{\rm WF},
a}]=Q$, which contradicts with $\mathbb{E}[q_i^{{\rm WF}, a}]>Q$
shown in the above under the presumption that $\mu_a< \mu_p$, it
thus concludes that $\mu_a\geq \mu_p$.

From the above discussions, we know that an alternative approach is
needed for comparing $C^{{\rm ER},a}_{\rm PR,WF}$ and $C^{{\rm
ER},p}_{\rm PR,WF}$. The result for this comparison and its proof
are given below:

\begin{theorem}\label{theorem:WF}
With the WF power control for the PR, $C^{{\rm ER},a}_{\rm
PR,WF}\geq C^{{\rm ER},p}_{\rm PR,WF}$, under the same average and
peak power threshold $\Gamma$.
\end{theorem}
\begin{proof}
The proof is based on the Lagrange duality of convex optimization
\cite{Boydbook}. First, we rewrite $C^{{\rm ER},a}_{\rm PR,WF}$ and
$C^{{\rm ER},p}_{\rm PR,WF}$ as the optimal values of the following
min-max optimization problems:
\begin{eqnarray}\label{eq:min max AIP}
C^{{\rm ER},a}_{\rm PR,WF}=\min_{\mu:\mu\geq 0}\max_{\{q_i\}:q_i\geq
0, \forall i}
\mathbb{E}\left[\log\left(1+\frac{f_iq_i}{1+I_i}\right)\right]
-\mu(\mathbb{E}[q_i]-Q)
\end{eqnarray}
and
\begin{eqnarray}\label{eq:min max PIP}
C^{{\rm ER},p}_{\rm PR,WF}=\min_{\mu:\mu\geq 0}\max_{\{q_i\}:q_i\geq
0, \forall i}
\mathbb{E}\left[\log\left(1+\frac{f_iq_i}{1+\Gamma}\right)\right]
-\mu(\mathbb{E}[q_i]-Q),
\end{eqnarray}
respectively. Note that $\mu_a$ and $\{q_i^{{\rm WF}, a}\}$ are the
optimal solutions to the ``min'' and ``max'' problems in
(\ref{eq:min max AIP}), respectively, and $\mu_p$ and $\{q_i^{{\rm
WF}, p}\}$ are the optimal solutions to the ``min'' and ``max''
problems in (\ref{eq:min max PIP}), respectively. Then, we have the
following equalities/inequalities:
\begin{eqnarray}
C^{{\rm ER},p}_{\rm PR,WF}&=&\min_{\mu:\mu\geq
0}\mathbb{E}\left[\left(\log\left(\frac{f_i}{(1+\Gamma)\mu}\right)\right)^+\right]-
\mathbb{E}\left[\left(1-\frac{(1+\Gamma)\mu}{f_i}\right)^+\right]+\mu
Q \label{eq:WF 1}
\\ &\leq& \mathbb{E}\left[\left(\log\left(\frac{f_i}{(1+\Gamma)\mu_a}\right)\right)^+\right]-
\mathbb{E}\left[\left(1-\frac{(1+\Gamma)\mu_a}{f_i}\right)^+\right]+\mu_aQ
\label{eq:WF 2}
\\ &=& \mathbb{E}_f\left[\left(\log\left(\frac{f_i}{(1+\mathbb{E}[I_i])\mu_a}\right)\right)^+-
\left(1-\frac{(1+\mathbb{E}[I_i])\mu_a}{f_i}\right)^+\right]+\mu_a Q
\label{eq:WF 3}
\\ &\leq&
\mathbb{E}_f\mathbb{E}_I\left[\left(\log\left(\frac{f_i}{(1+I_i)\mu_a}\right)\right)^+-
\left(1-\frac{(1+I_i)\mu_a}{f_i}\right)^+\right]+\mu_aQ \label{eq:WF
4} \\ &=& \min_{\mu:\mu\geq
0}\mathbb{E}\left[\left(\log\left(\frac{f_i}{(1+I_i)\mu}\right)\right)^+\right]-
\mathbb{E}\left[\left(1-\frac{(1+I_i)\mu}{f_i}\right)^+\right]+\mu Q
\label{eq:WF 5}
\\ &=& C^{{\rm ER},a}_{\rm PR,WF} \ , \label{eq:WF 6}
\end{eqnarray}
where (\ref{eq:WF 1}) is obtained by substituting $\{q_i^{{\rm WF},
p}\}$ in (\ref{eq:WF PIP}) with $\mu_p$ replaced by an arbitrary
positive $\mu$ into (\ref{eq:min max PIP}); (\ref{eq:WF 2}) is due
to the fact that $\mu_a$ is not the minimizer $\mu_p$ for
(\ref{eq:WF 1}); (\ref{eq:WF 3}) is due to $\mathbb{E}[I_i]=\Gamma$;
(\ref{eq:WF 4}) is due to convexity of the function in
$\mathbb{E}_f[\cdot]$ of (\ref{eq:WF 3}) with respect to
$\mathbb{E}[I_i]$ for any given $f_i$ and Jensen's inequality;
(\ref{eq:WF 5}) and (\ref{eq:WF 6}) are due to the fact that $\mu_a$
and $\{q_i^{{\rm WF}, a}\}$ in (\ref{eq:WF AIP}) are the optimal
solutions to the min-max optimization problem in (\ref{eq:min max
AIP}).
\end{proof}

Theorem \ref{theorem:WF} suggests that, similarly like the CP
policy, under the WF-based power control, randomized interference
power levels due to the CR transmission in the AIP case is superior
over constant interference power levels in the PIP case in terms of
the maximum achievable PR ergodic capacity. However, the
interference diversity gain observed here is not as obvious as that
in the CP case due to the more complex WF-based PR power control.

\subsection{Outage Capacity}

\subsubsection{Constant-Power Policy}

With the CP policy in (\ref{eq:CP}), for a given outage probability,
$\epsilon_0$, the maximum achievable constant SNR at PR-Rx, denoted
as $\gamma_a$, in the AIP case can be obtained from
$Pr\{(f_iQ)/(1+I_i)< \gamma_a \}=\epsilon_0$, and the corresponding
outage capacity, denoted as $C_{\rm PR, CP}^{{\rm
OUT},a}(\epsilon_0)$, is equal to $\log(1+\gamma_a)$. Similarly, for
the same $\epsilon_0$, the maximum achievable constant SNR at PR-Rx,
$\gamma_p$, in the PIP case can be obtained from
$Pr\{(f_iQ)/(1+\Gamma)< \gamma_p \}=\epsilon_0$, and the
corresponding outage capacity, $C_{\rm PR, CP}^{{\rm
OUT},p}(\epsilon_0)$, is obtained as $\log(1+\gamma_p)$.

Instead of comparing $C_{\rm PR, CP}^{{\rm OUT},a}(\epsilon_0)$ and
$C_{\rm PR, CP}^{{\rm OUT},p}(\epsilon_0)$ directly, we consider the
following equivalent problem: Supposing that
$\gamma_a=\gamma_p=\gamma_0$, we compare the resultant minimum
outage probabilities, denoted as $\epsilon_a$ and $\epsilon_p$ in
the AIP and PIP cases, respectively. If $\epsilon_a\leq \epsilon_p$
for any given $\gamma_0$, we conclude that $C_{\rm PR, CP}^{{\rm
OUT},a}(\epsilon_0)\geq C_{\rm PR, CP}^{{\rm OUT},p}(\epsilon_0)$
for any $\epsilon_0$. This is true because if $\epsilon_a\leq
\epsilon_p$, we can increase $\gamma_a$ above $\gamma_0$ so that
$\epsilon_a$ increases until it becomes equal to $\epsilon_p$; since
$\gamma_a\geq \gamma_0\geq \gamma_p$, it follows that $C_{\rm PR,
CP}^{{\rm OUT},a}(\epsilon_p)\geq C_{\rm PR, CP}^{{\rm
OUT},p}(\epsilon_p)$. Similarly, if $\epsilon_a\geq\epsilon_p$ for
any given $\gamma_0$, we conclude that $C_{\rm PR, CP}^{{\rm
OUT},a}(\epsilon_0)\leq C_{\rm PR, CP}^{{\rm OUT},p}(\epsilon_0)$
for any $\epsilon_0$.

To compare $\epsilon_a$ and $\epsilon_p$ for the same given
$\gamma_0$, we first express $\epsilon_a$ as
\begin{eqnarray}
\epsilon_a&=&Pr\left\{\frac{f_iQ}{1+I_i}< \gamma_0 \right\} \\
&=& \mathbb{E}_I\left[\mathbb{E}_f\left[{\bf 1}
\left(\frac{f_iQ}{1+I_i}< \gamma_0\right)\right]\right] \\ &=&
\mathbb{E}_I\left[G_f\left(\frac{(1+I_i)\gamma_0}{Q}\right)\right]
\label{eq:outage prob AIP}
\end{eqnarray}
where $G_f(x)$ is the cumulative density function (CDF) for
$\mv{f}$, i.e., $G_f(x)=Pr\{\mv{f}<x\}$. Similarly, we can express
$\epsilon_p$ as
\begin{eqnarray}
\epsilon_p=G_f\left(\frac{(1+\Gamma)\gamma_0}{Q}\right).
\label{eq:outage prob PIP}
\end{eqnarray}
By Jensen's inequality, from (\ref{eq:outage prob AIP}) and
(\ref{eq:outage prob PIP}), it follows that $\epsilon_a\leq
\epsilon_p$ if $G_f(x)$ is a convex function. Similarly,
$\epsilon_a\geq \epsilon_p$ if $G_f(x)$ is a concave function. We
thus have the following theorem:
\begin{theorem}\label{theorem:CP Outage}
With the CP policy for the PR, $C^{{\rm OUT},a}_{\rm
PR,CP}(\epsilon_0)\geq C^{{\rm OUT},p}_{\rm PR,CP}(\epsilon_0),
\forall \epsilon_0$, under the same average and peak power threshold
$\Gamma$, if $G_f(x)$ is a convex function; and $C^{{\rm
OUT},a}_{\rm PR,CP}(\epsilon_0)\leq C^{{\rm OUT},p}_{\rm
PR,CP}(\epsilon_0), \forall \epsilon_0$, if $G_f(x)$ is a concave
function.
\end{theorem}

Theorem \ref{theorem:CP Outage} suggests that for the CP policy,
whether the AIP or the PIP constraint results in a larger PR outage
capacity depends on the convexity/concavity of the CDF of the PR
fading channel power gain. As an example, for the standard Rayleigh
fading model, it is known that $G_f(x)$ has an exponential
distribution that is convex and, thus, $C^{{\rm OUT},a}_{\rm
PR,CP}(\epsilon_0)\geq C^{{\rm OUT},p}_{\rm PR,CP}(\epsilon_0)$.
However, in general, whether the interference diversity gain is
present depends on the PR fading channel distribution.

\subsubsection{Channel-Inversion Power Control}

Next, we consider the special case of the CR outage capacity with
zero outage probability, i.e., the delay-limited capacity, which is
achievable by the CI power-control policy. In the AIP case, the
optimal PR power allocation is expressed as
\begin{equation}\label{eq:CI AIP}
q_i^{{\rm CI},a}=\frac{\gamma_a(1+I_i)}{f_i}
\end{equation}
and in the PIP case expressed as
\begin{equation}\label{eq:CI PIP}
q_i^{{\rm CI},p}=\frac{\gamma_p(1+\Gamma)}{f_i}
\end{equation}
where $\gamma_a$ and $\gamma_p$ are the constant SNRs at PR-Rx for
the AIP and PIP cases, respectively. Given $\mathbb{E}[q_i]=Q$,
$\gamma_a$ and $\gamma_p$ can be obtained from (\ref{eq:CI AIP}) and
(\ref{eq:CI PIP}) as
\begin{equation}\label{eq:SINR AIP}
\gamma_a=\frac{Q}{\mathbb{E}\left[\frac{1+I_i}{f_i}\right]}
\end{equation}
and
\begin{equation}\label{eq:SINR PIP}
\gamma_p=\frac{Q}{(1+\Gamma)\mathbb{E}\left[\frac{1}{f_i}\right]},
\end{equation}
respectively. Since $f_i$ is independent of $I_i$, we have
\begin{equation}
\mathbb{E}\left[\frac{1+I_i}{f_i}\right]=\mathbb{E}_f\left[\frac{1+\mathbb{E}[I_i]}{f_i}\right]
=(1+\Gamma)\mathbb{E}\left[\frac{1}{f_i}\right] \nonumber
\end{equation}
and thus it follows from (\ref{eq:SINR AIP}) and (\ref{eq:SINR PIP})
that $\gamma_a=\gamma_p$. Hence, we conclude that the PR
delay-limited capacities, expressed as $C^{{\rm DL},a}_{\rm
PR}=\log(1+\gamma_a)$ and $C^{{\rm DL},p}_{\rm
PR}=\log(1+\gamma_p)$, for the AIP and PIP cases, respectively, are
indeed identical. The following theorem thus holds:
\begin{theorem}\label{theorem:CI}
With the CI power control for the PR, $C^{{\rm DL},a}_{\rm PR}=
C^{{\rm DL},p}_{\rm PR}$, under the same average and peak power
threshold $\Gamma$.
\end{theorem}

Theorem \ref{theorem:CI} suggests that for the CI power control, the
loss of the PR delay-limited capacity due to randomized interference
powers from CR-Tx is identical to that due to constant interference
powers, i.e., the AIP constraint is at least no worse than the PIP
from the PR's perspective of delivering zero-delay and constant-rate
data traffic.

\subsubsection{Truncated-Channel-Inversion Power Control}

Lastly, we consider the general outage capacity for the PR
achievable by the TCI power-control policy. In the AIP case, the
optimal TCI power control is expressed as
\begin{equation}\label{eq:TCI AIP}
q_i^{{\rm TCI},a}=\left\{\begin{array}{ll}
\frac{\gamma_a(1+I_i)}{f_i}, & \frac{f_i}{1+I_i}\geq \theta_a \\ 0,
& {\rm otherwise}\end{array} \right.
\end{equation}
where $\theta_a$ is the threshold for the effective channel power
gain above which CI power control is applied to achieve a constant
receiver SNR, $\gamma_a$, and below which no transmission is
implemented. Similarly, the TCI power control in the PIP case is
expressed as
\begin{equation}\label{eq:TCI PIP}
q_i^{{\rm TCI},p}=\left\{\begin{array}{ll}
\frac{\gamma_p(1+\Gamma)}{f_i}, & \frac{f_i}{1+\Gamma}\geq \theta_p \\
0, & {\rm otherwise}\end{array} \right.
\end{equation}
where $\theta_p$ is the threshold for power truncation. Given the
outage probability $\epsilon_0$, $\theta_a$ and $\theta_p$ can be
obtained from $Pr\{f_i/(1+I_i)< \theta_a\}=\epsilon_0$ and
$Pr\{f_i/(1+\Gamma)< \theta_p\}=\epsilon_0$, respectively. Then,
$\gamma_a$ and $\gamma_p$ can be obtained from $\mathbb{E}[q_i^{{\rm
TCI},a}]=Q$ and $\mathbb{E}[q_i^{{\rm TCI},p}]=Q$, respectively. The
corresponding outage capacities for the PR, denoted as $C^{{\rm
OUT},a}_{\rm PR,TCI}(\epsilon_0)$ and $C^{{\rm OUT},p}_{\rm
PR,TCI}(\epsilon_0)$ for the AIP and PIP cases can be obtained as
$\log(1+\gamma_a)$ and $\log(1+\gamma_p)$, respectively.

\begin{theorem}\label{theorem:TCI}
With the TCI power control for the PR, $C^{{\rm OUT},a}_{\rm
PR,TCI}(\epsilon_0)\geq C^{{\rm OUT},p}_{\rm PR,TCI}(\epsilon_0),
\forall \epsilon_0$, under the same average and peak power threshold
$\Gamma$.
\end{theorem}
\begin{proof}
Similarly like the discussions for the PR outage capacity with the
CP policy, we compare $C^{{\rm OUT},a}_{\rm PR,TCI}(\epsilon_0)$ and
$C^{{\rm OUT},p}_{\rm PR,TCI}(\epsilon_0)$ via the following
equivalent problem: Given $\gamma_a=\gamma_p=\gamma_0$, we compare
the minimum outage probabilities in the AIP and PIP cases, denoted
as $\epsilon_a$ and $\epsilon_p$, respectively. If $\epsilon_a\leq
\epsilon_p, \forall \gamma_0$, we then conclude that $C^{{\rm
OUT},a}_{\rm PR,TCI}(\epsilon_0)\geq C^{{\rm OUT},p}_{\rm
PR,TCI}(\epsilon_0),\forall \epsilon_0$.

Next, we show that $\epsilon_a\leq \epsilon_p, \forall \gamma_0$.
Similarly like the proof of Theorem \ref{theorem:WF}, the Lagrange
duality is applied here. For given $Q$ and $\gamma_0$, $\epsilon_a$
and $\epsilon_p$ can be rewritten as the optimal values of the
following max-min optimization problems:
\begin{eqnarray}\label{eq:max min AIP}
\epsilon_a=\max_{\mu:\mu\geq 0}\min_{\{q_i\}:q_i\geq 0, \forall i}
Pr\left\{\frac{f_iq_i}{1+I_i}<\gamma_0\right\}
+\mu(\mathbb{E}[q_i]-Q)
\end{eqnarray}
and
\begin{eqnarray}\label{eq:max min PIP}
\epsilon_p=\max_{\mu:\mu\geq 0}\min_{\{q_i\}:q_i\geq 0, \forall i}
Pr\left\{\frac{f_iq_i}{1+\Gamma}<\gamma_0\right\}
+\mu(\mathbb{E}[q_i]-Q),
\end{eqnarray}
respectively. Note that $\mu_a=\theta_a/\gamma_0$ and $\{q_i^{{\rm
TCI}, a}\}$ are the optimal solutions to the ``max'' and ``min''
problems in (\ref{eq:max min AIP}), respectively, and
$\mu_p=\theta_p/\gamma_0$ and $\{q_i^{{\rm TCP}, p}\}$ are the
optimal solutions to the ``max'' and ``min'' problems in
(\ref{eq:max min PIP}), respectively. Then, we have the following
equalities/inequalities:
\begin{eqnarray}
\epsilon_p&=&\max_{\mu:\mu\geq 0} \mathbb{E}_f\left[{\bf
1}\left(\frac{f_i}{1+\Gamma}<\gamma_0\mu \right)
\right]+\mu\mathbb{E}_f\left[\frac{(1+\Gamma)\gamma_0}{f_i}{\bf1}\left(\frac{f_i}{1+\Gamma}\geq\gamma_0\mu
\right)\right]-\mu Q \label{eq:TCI 1}
\\ &\geq& \mathbb{E}_f\left[{\bf
1}\left(\frac{f_i}{1+\Gamma}<\gamma_0\mu_a \right)
\right]+\mu_a\mathbb{E}_f\left[\frac{(1+\Gamma)\gamma_0}{f_i}{\bf1}\left(\frac{f_i}{1+\Gamma}\geq\gamma_0\mu_a
\right)\right]-\mu_aQ \label{eq:TCI 2}
\\ &=&
1+\mathbb{E}_f\left[\left(\frac{(1+\Gamma)\gamma_0\mu_a}{f_i}-1\right){\bf1}\left(\frac{f_i}{1+\Gamma}\geq\gamma_0\mu_a
\right)\right]-\mu_a Q \label{eq:TCI 3}
\\ &=&
1+\mathbb{E}_f\left[\left(\frac{(1+\mathbb{E}[I_i])\gamma_0\mu_a}{f_i}-1\right){\bf1}\left(\frac{f_i}{1+\mathbb{E}[I_i]}\geq\gamma_0\mu_a
\right)\right]-\mu_a Q \label{eq:TCI 4}
\\ &\geq& 1+\mathbb{E}_f\mathbb{E}_I\left[\left(\frac{(1+I_i)\gamma_0\mu_a}{f_i}-1\right){\bf1}\left(\frac{f_i}{1+I_i}\geq\gamma_0\mu_a
\right)\right]-\mu_a Q \label{eq:TCI 5}
\\ &=& \epsilon_a \label{eq:TCI 6}
\end{eqnarray}
where (\ref{eq:TCI 1}) is obtained by substituting $\{q_i^{{\rm
TCI}, p}\}$ in (\ref{eq:TCI PIP}) with $\theta_p$ replaced by
$\gamma_0\mu$ into (\ref{eq:max min PIP}); (\ref{eq:TCI 2}) is due
to the fact that $\mu_a$ is not the maximizer $\mu_p$ for
(\ref{eq:TCI 1}); (\ref{eq:TCI 4}) is due to
$\mathbb{E}[I_i]=\Gamma$; (\ref{eq:TCI 5}) is due to concavity of
the function in $\mathbb{E}_f[\cdot]$ of (\ref{eq:TCI 4}) with
respect to $\mathbb{E}[I_i]$ for any given $f_i$ and Jensen's
inequality; (\ref{eq:TCI 6}) is due to the fact that $\mu_a$ and
$\{q_i^{{\rm TCI}, a}\}$ in (\ref{eq:TCI AIP}) are the optimal
solutions to the the max-min optimization problem in (\ref{eq:max
min AIP}).
\end{proof}

Theorem \ref{theorem:TCI} suggests that for the TCI power control of
the PR, the interference diversity gain due to the AIP constraint
over the PIP exists regardless of the outage probability. Note that
in Theorem \ref{theorem:CI} for the extreme case of zero outage
probability, it has been shown that the delay-limited capacities are
the same for both the AIP and PIP constraints.

\section{Simulation Results and Discussions}\label{sec:simulation
results}

So far, we have studied the effect of the AIP and PIP constraints on
the ergodic/outage capacity of the CR link and the PR link
separately. In this section, we will consider a realistic spectrum
sharing scenario over the fading channels, and evaluate by
simulation the jointly achievable ergodic/outage capacities for both
the PR and CR links. In total, we will consider four cases of
different combinations, which are CR ergodic capacity versus PR
ergodic capacity, CR ergodic capacity versus PR outage capacity, CR
outage capacity versus PR ergodic capacity, and CR outage capacity
versus PR outage capacity, in Figs. \ref{fig:ER ER}-\ref{fig:OUT
OUT}, respectively. It is assumed that $\Gamma_a=\Gamma_p=1$, the
same as the additive Gaussian noise power at PR-Rx and CR-Rx. It is
also assumed that $\mv{h}$, $\mv{g}$, and $\mv{f}$ are obtained from
the Rayleigh fading model, i.e., they are the squared norms of
independent CSCG r.v.s $\sim$ $\mathcal{CN}(0,1)$,
$\mathcal{CN}(0,10)$, and $\mathcal{CN}(0,1)$, respectively. Note
that we have purposely set the average power for $\mv{g}$ to be 10
dB larger than that for $\mv{h}$ or $\mv{f}$ so as to pronounce the
effect of the interference channel from CR-Tx to PR-Rx on the
achievable capacities. The PR transmit power constraint is set to be
$Q=10$. In the cases of outage capacities of the PR and/or CR, the
outage probability targets of $\epsilon_0$ for PR and CR are set to
be 0.2 and 0.1, respectively. In each figure, the PR and CR
capacities in bits/complex dimension (dim.) are plotted versus the
additional channel power gain attenuation of $\mv{g}$ in dB. For
example, for 0-dB attenuation, $\mathbb{E}[g_i]=10$; for 10-dB
attenuation, $\mathbb{E}[g_i]=1$.

In Figs. \ref{fig:ER ER} and \ref{fig:ER OUT}, we compare the CR
ergodic capacity under AIP or PIP constraint with the corresponding
ergodic and outage capacities for the PR, respectively. Note that
the CR ergodic capacities shown in these two figures are the same.
With increasing channel attenuation of $\mv{g}$, it is observed that
the CR ergodic capacity increases for both AIP and PIP cases. This
is obvious since given the fixed peak or average interference-power
threshold at PR-Rx, decreasing of the average power for $\mv{g}$
results in increasing of the average transmit power of the CR. It is
also observed that the AIP-based optimal power control performs
better than the PIP-based one for the CR, since the former is more
flexible for exploiting all the available CSI at CR-Tx.
Interestingly, as the average power for $\mv{g}$ decreases,
eventually the CR ergodic capacities in the AIP and PIP cases
converge to the same value. This can be explained as follows. From
(\ref{eq:power AIP ER}) and (\ref{eq:power PIP}), it follows that in
the AIP case, the interference power at PR-Rx, $I_i$, is randomized
over $i$ (but with $\mathbb{E}[I_i]=\Gamma$), while in the PIP case,
$I_i$ is constantly equal to $\Gamma$ for each $i$. Note that the
above fact leads to the interference diversity gain of the AIP over
the PIP for the PR transmission. However, with $g_i\rightarrow 0$,
it can be shown in the AIP case that $1/\nu\rightarrow \Gamma$ and
$I_i\rightarrow \Gamma$, which implies that $p_i^{{\rm
ER},a}=p_i^{{\rm ER},p}=\Gamma/g_i, \forall i$, and thus the same CR
ergodic capacity is resultant for both the AIP and PIP cases.

On the other hand, it is observed that the ergodic and outage
capacities for the PR under the AIP from CR-Tx are larger than the
corresponding ones under the PIP for various PR power-control
policies, which are in accord with the analytical results obtained
in Section \ref{sec:primary user capacity}. Note that the PR
ergodic/outage capacities in the PIP case are fixed regardless of
the channel power for $\mv{g}$, since $I_i$ is fixed as $\Gamma$ at
PR-Rx for each $i$. However, the PR ergodic/outage capacities in the
AIP case are observed to decrease with increasing of the channel
attenuation of $\mv{g}$. This is due to the fact that, as explained
earlier, $I_i\rightarrow \Gamma$ as $g_i\rightarrow 0$. Since the
capacity gain of the AIP over the PIP is due to the randomness of
$I_i$ over $i$ , this interference diversity gain diminishes as
$I_i\rightarrow \Gamma, \forall i$.

In Figs. \ref{fig:OUT ER} and \ref{fig:OUT OUT}, we compare the CR
outage capacity under AIP or PIP constraint with the corresponding
ergodic and outage capacities for the PR, respectively. Note that
the CR outage capacities shown in these two figures are identical.
With increasing channel attenuation of $\mv{g}$, it is observed
that, as expected, the CR outage capacity increases for both AIP and
PIP cases. It is also observed that the AIP-based optimal power
control results in substantial outage capacity gains than the
PIP-based one for the CR. It can be shown that as the average power
for $\mv{g}$ decreases, eventually the CR outage capacity gaps
between the AIP and the PIP cases converge to
$\log(\zeta_a/\zeta_p)$ for a given $\epsilon_0$. The proof is given
as follows. Suppose that $g'_i=\kappa g_i, \forall i$, where
$\kappa$ is a positive constant; we thus have
$\mathbb{E}[g'_i]=\kappa\mathbb{E}[g_i]$. For a given $\epsilon_0$,
it then follows that the new value of threshold in (\ref{eq:power
AIP OUT}) becomes $\lambda'=\lambda/\kappa$. From (\ref{eq:power AIP
OUT}) and under the same value of $\Gamma_a$, we have
$\zeta'_a=\zeta_a/\kappa$. Thus, the outage capacity corresponding
to $\mv{g}'$ in the AIP case is expressed as
$\log(1+\zeta_a/\kappa)$. Similarly, we can show that in the PIP
case, the new value of $\zeta_p$ corresponding to $\mv{g}'$ is
$\zeta'_p=\zeta_p/\kappa$ and thus the corresponding outage capacity
becomes $\log(1+\zeta_p/\kappa)$. Thus, the outage capacity gap
between the AIP and PIP cases is equal to
$\log(\frac{1+\zeta_a/\kappa}{1+\zeta_p/\kappa})$. As
$\kappa\rightarrow 0$, we conclude that the above capacity gap
converges to $\log(\zeta_a/\zeta_p)$. Note that in this simulation
with $\epsilon_0=0.1$ for the CR, $\log(\zeta_a/\zeta_p)=2.6791$
bits/complex dimension.

Furthermore, it is observed that the ergodic and outage capacities
for the PR under the AIP from CR-Tx are also larger than the
corresponding ones under the PIP for various PR power-control
policies, as have been analytically shown in Section
\ref{sec:primary user capacity}. Note that not only the PR
ergodic/outage capacities in the PIP case are fixed for all the
average powers for $\mv{g}$ due to that $I_i$ is fixed as $\Gamma$
for each $i$, but also are these capacities in the AIP case. The
latter observation can be explained by noting from the earlier proof
that for any channel power gains $g'_i$, $g'_i=\kappa g_i, \forall
i$, the resultant interference power at PR-Rx, $I'_i$, can be shown
to have the same distribution as $I_i$; as a result, the PR
capacities are constant regardless of $\kappa$.

\section{Concluding Remarks} \label{sec:conclusion}

This paper studies the information-theoretic limits for wireless
spectrum sharing in the PR-CR network where the CR applies the
interference-power/interference-temperature constraint at the PR
receiver as a practical means to protect the PR transmission. On the
contrary to the traditional viewpoint that the
peak-interference-power (PIP) constraint protects better the PR
transmission than the average-interference-power (AIP) constraint
given their same power-threshold value, this paper shows that the
AIP constraint can be in many cases more advantageous over the PIP
for minimizing the resultant capacity losses of the PR fading
channel. This is mainly owing to an interesting interference
diversity phenomenon discovered in this paper. This paper thus
provides an important design rule for the CR networks in practice,
i.e., the AIP constraint should be used for the purposes of both
protecting the PR transmission as well as maximizing the CR
throughput.

This paper assumes that the perfect CSI on the interference channel
from the CR transmitter to the PR receiver is available at the CR
transmitter for each fading state. In practice, it is usually more
valid to assume availability of only the statistical channel
knowledge. The definition of the AIP constraint in this paper can be
extendible to such cases. Furthermore, this paper considers the
fading PR and CR channels, but more generally, the results obtained
also apply to other channel models consisting of parallel Gaussian
channels over which the average and peak power constraints are
applicable, e.g., the time-dispersive broadband channel that is
decomposable into parallel narrow-band channels by the well-known
orthogonal-frequency-division-multiplexing (OFDM)
modulation/demodulation.

\newpage

\begin{figure}
\centering{
 \epsfxsize=4.6in
    \leavevmode{\epsfbox{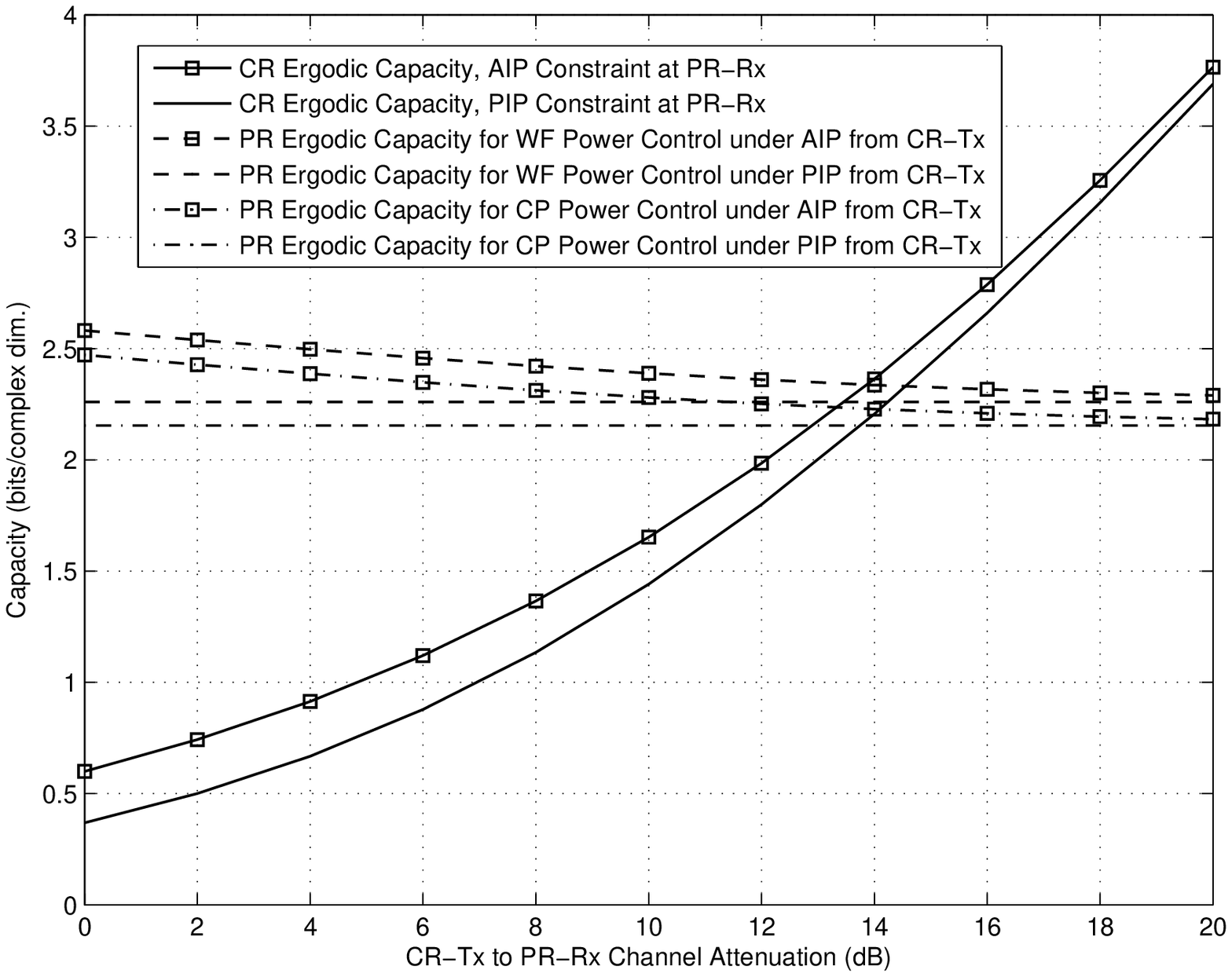}} }
\caption{Jointly achievable CR ergodic capacity and PR ergodic
capacity.}\label{fig:ER ER}
\end{figure}

\begin{figure}
\centering{
 \epsfxsize=4.6in
    \leavevmode{\epsfbox{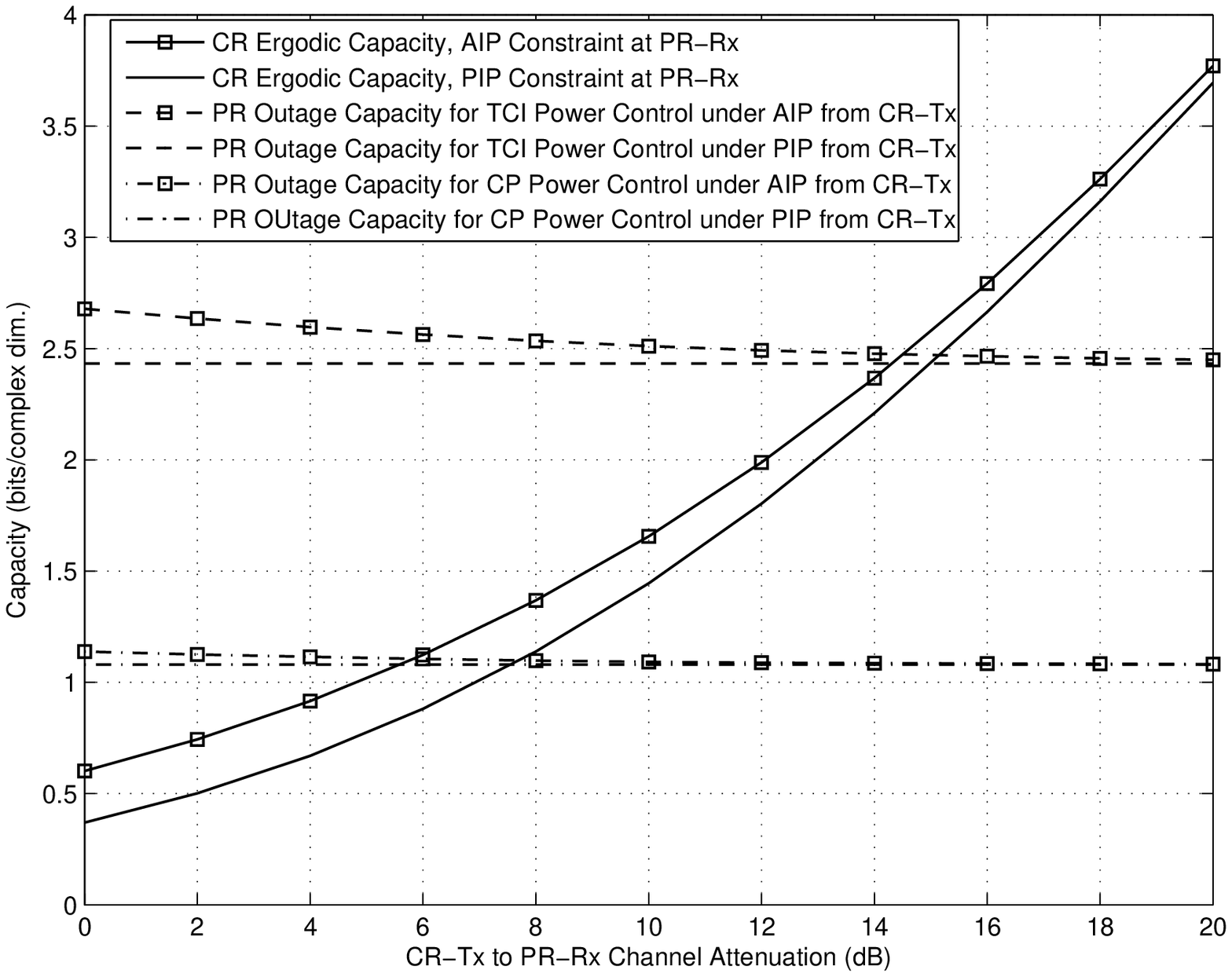}} }
\caption{Jointly achievable CR ergodic capacity and PR outage
capacity.}\label{fig:ER OUT}
\end{figure}

\begin{figure}
\centering{
 \epsfxsize=4.6in
    \leavevmode{\epsfbox{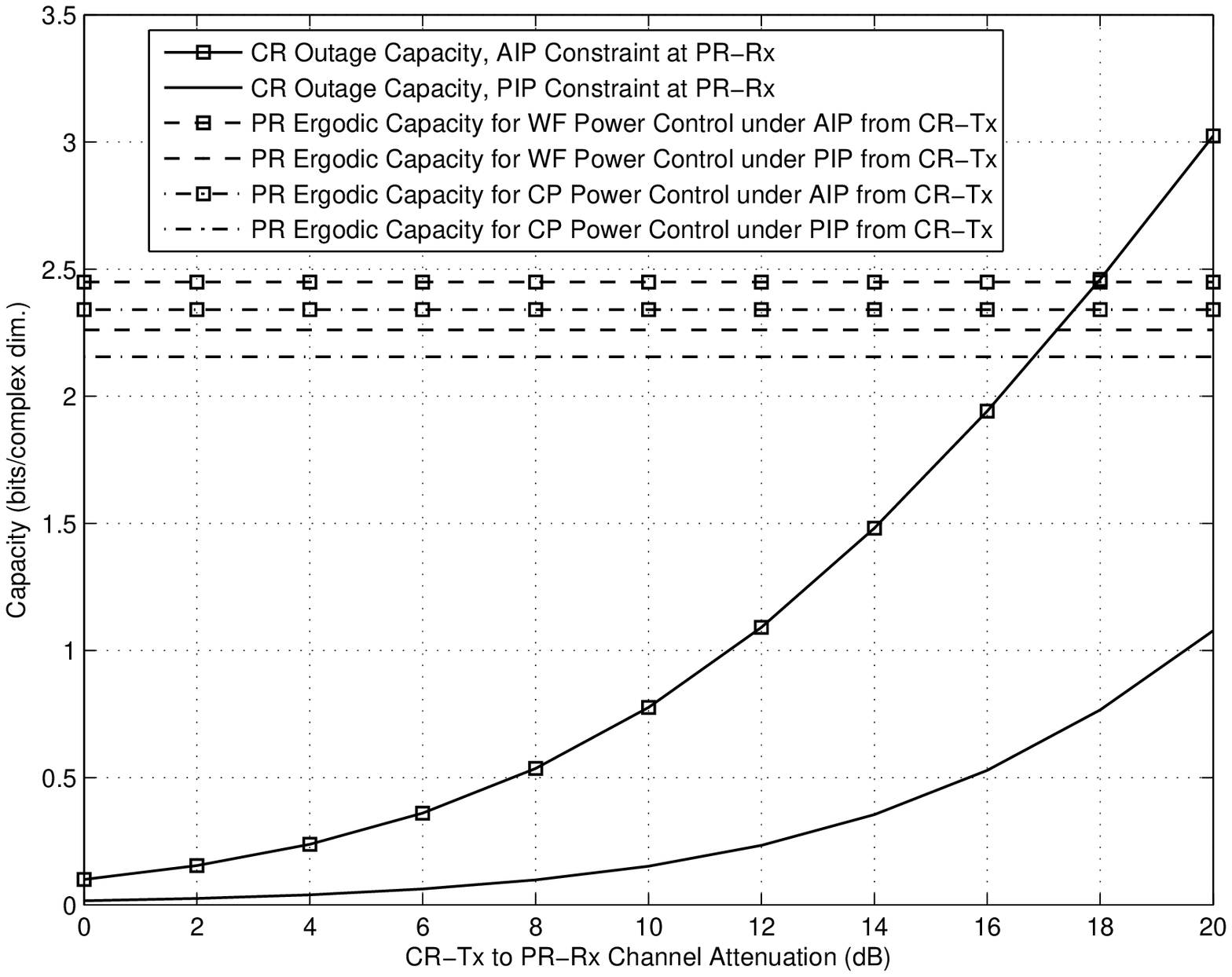}} }
\caption{Jointly achievable CR outage capacity and PR ergodic
capacity.}\label{fig:OUT ER}
\end{figure}

\begin{figure}
\centering{
 \epsfxsize=4.6in
    \leavevmode{\epsfbox{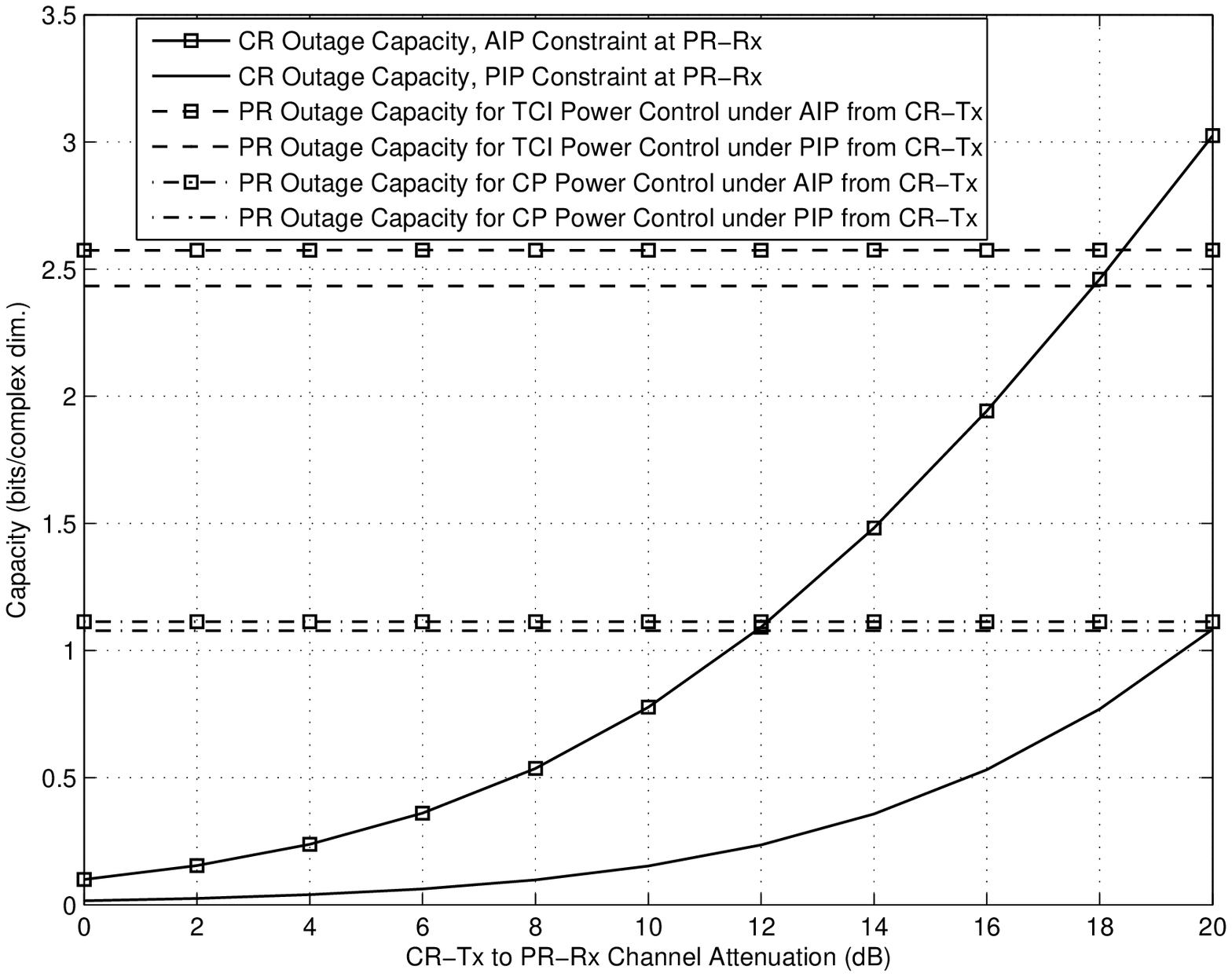}} }
\caption{Jointly achievable CR outage capacity and PR outage
capacity.}\label{fig:OUT OUT}
\end{figure}

\end{document}